\documentclass[sigconf]{acmart}

\usepackage{citehack}
\usepackage{graphicx}
\usepackage{amsmath}
\usepackage{times}
\usepackage{url}
\usepackage{caption}

\usepackage{amssymb}
\usepackage{bm}
\usepackage{nicefrac}
\usepackage{booktabs}
\usepackage{array}
\usepackage{multirow}
\usepackage{threeparttable}
\usepackage{makecell}
\usepackage[procnumbered,ruled,vlined,linesnumbered]{algorithm2e}

\allowdisplaybreaks[3]

\def\mo#1{\| #1 \|}\def\kh#1{\left( #1 \right)}

\def\mo#1{\| #1 \|}
\def\len#1{\left| #1 \right|}

\def\calG{\mathcal{G}}
\def\calP{\mathcal{P}}

\newcommand{\rea}{\mathbb{R}}

\newcommand\yy{\boldsymbol{\mathit{y}}}
\newcommand\zz{\boldsymbol{\mathit{z}}}
\newcommand\xx{\boldsymbol{\mathit{x}}}
\newcommand\bb{\boldsymbol{\mathit{b}}}
\newcommand\ee{\boldsymbol{\mathit{e}}}
\newcommand\qq{\boldsymbol{\mathit{q}}}
\newcommand\pp{\boldsymbol{\mathit{p}}}
\newcommand\uu{\boldsymbol{\mathit{u}}}
\newcommand\vs{\boldsymbol{\mathit{s\emph{}}}}
\newcommand{\eps}{\epsilon}

\newcommand\LL{\bm{\mathit{L}}}
\newcommand\AAA{\boldsymbol{\mathit{A}}}
\newcommand\BB{\boldsymbol{\mathit{B}}}

\newcommand\DD{\boldsymbol{\mathit{D}}}

\newcommand\TT{\boldsymbol{\mathit{T}}}

\newcommand\II{\boldsymbol{\mathit{I}}}
\newcommand\WW{\boldsymbol{\mathit{W}}}

\newcommand\UU{\boldsymbol{\mathit{U}}}
\newcommand\XX{\boldsymbol{\mathit{X}}}
\newcommand\YY{\boldsymbol{\mathit{Y}}}

\newtheorem{definition}{Definition}[section]
\newtheorem{lemma}{lemma}[section]
\newtheorem{theorem}{theorem}[section]
\newtheorem{proposition}{Proposition}
\DontPrintSemicolon
\SetKw{KwAnd}{and}
\SetProcNameSty{textsc}
\SetFuncSty{textsc}
\SetKwInOut{Input}{Input\ \ \ \ }
\SetKwInOut{Output}{Output}

\def\kh#1{\left( #1 \right)}

{\begin{Sbox}\begin{minipage}}%
		{\end{minipage}\end{Sbox}\fbox{\TheSbox}}

\AtBeginDocument{%
	\providecommand\BibTeX{{%
			\normalfont B\kern-0.5em{\scshape i\kern-0.25em b}\kern-0.8em\TeX}}}

\setcopyright{acmcopyright}
\copyrightyear{2021}
\acmYear{2021}
\acmDOI{10.1145/1122445.1122456}

\acmConference[WWW 2021]{WWW 2021: (IW3C2) International World Wide Web Conference Committee}{April 19--23, 2021}{Ljubljana, Slovenia}
\acmBooktitle{WWW 2021: (IW3C2) International World Wide Web Conference Committee, April 19--23, 2021, Ljubljana, Slovenia}
\acmPrice{15.00}
\acmISBN{978-1-4503-XXXX-X/18/06}

\begin{document}

\title{Fast Evaluation for Relevant Quantities of Opinion Dynamics}

\author{Wanyue Xu}
\affiliation{%
	\institution{Fudan University}
	\city{Shanghai}
	\country{China}}
\email{xuwy@fudan.edu.cn}

\author{Qi Bao}
\affiliation{%
	\institution{Fudan University}
	\city{Shanghai}
	\country{China}}
\email{18210240001@fudan.edu.cn}

\author{Zhongzhi Zhang}
\affiliation{%
	\institution{Fudan University}
	\city{Shanghai}
	\country{China}}
\email{zhangzz@fudan.edu.cn}

\renewcommand{\shortauthors}{Wanyue Xu, Qi Bao and Zhongzhi Zhang}


\begin{abstract}
	One of the main subjects in the field of social networks is to quantify conflict, disagreement, controversy, and polarization, and some quantitative indicators have been developed to quantify these concepts. However, direct computation of these indicators involves the operations of matrix inversion and multiplication, which make it computationally infeasible for large-scale graphs with millions of nodes. In this paper, by reducing the problem of computing relevant quantities to evaluating $\ell_2$ norms of some vectors, we present a nearly linear time algorithm to estimate all these quantities. Our algorithm is based on the Laplacian solvers, and has a proved theoretical  guarantee of  error for each quantity. We execute extensive numerical experiments on a variety of real networks, which demonstrate that our approximation algorithm is efficient and effective, scalable to large graphs having millions of nodes.
\end{abstract}



\maketitle

\section{Introduction}

Online social networks and social media are increasingly becoming an important part of our lives, which have led to a fundamental change of ways people share and shape opinions~\cite{JiMiFrBu15,DoZhKoDiLi18,AnYe19}. 
Particularly, the enormous popularity of social media and online social networks produces diverse social phenomena, such as polarization, disagreement, conflict, and controversy, which have been a hot subject of study in different disciplines, especially social science. In fact, some phenomena, for example, disagreement and polarization, have taken place in human societies for millenia, but now they are more apparent in an online virtual world.

In addition to the identification of aforementioned social phenomena, the issue of how to quantify these phenomena has received increasing amounts of attention. Thus far, various measures have been developed to quantify these phenomena, such as disagreement~\cite{MuMuTs18,DaGoLe13}, polarization~\cite{ DaGoLe13, MaTeTs17, MuMuTs18},  conflict~\cite{ChLiDe18}, and controversy~\cite{ChLiDe18}. Most of these measures are based on the Friedkin-Johnsen (FJ)  social-opinion dynamics model~\cite{FrJo90}, which is an important extension of the DeGroot's opinion model~\cite{De74}. Although the expressions of these quantitative metrics seem very concise, rigorous determination for them in large-scale graphs is a computational challenge, since it involves matrix inversion and multiplication.

In this paper, we address the problem of fast calculation for the aforementioned quantitative measures related to opinion dynamics modelled by the well-established FJ model, by exploiting the connection~\cite{MuMuTs18, ChLiDe18} between forest matrix~\cite{ChSh97,ChSh98} and these key quantities. To this end, we first represent these quantities in terms of the $\ell_2$ norm of some vectors. We then provide an algorithm to approximate these quantities in nearly linear time with the number of edges. Our algorithm has a proved error guarantee. Extensive experiments on many real network datasets indicate that our algorithm is efficient and effective, which is scalable to large graphs with millions of nodes.

\textbf{Related work.} The focus of this paper is to propose a fast algorithm approximately evaluating the quantitative metrics of some key social phenomena. We use the FJ model as our underlying opinion dynamics model. Below, we review some work that is closely related to ours.

It is well known, the FJ model is a significant extension of the DeGroot model for opinion dynamics, where opinion may be the understanding or position of individuals on a certain popular topic or subject. The DeGroot model is an iterative averaging model, with each individual having only one opinion~\cite{De74}. At each time step, any individual updates its opinion as a weighted average of its neighbors. For the DeGroot model on a connected graph, it will reach consensus~\cite{Be81}. Under the formalism of the DeGroot model or its variants, many consensus protocols have been proposed or studied~\cite{SaFaMu07}, especially in the literature of  systems~\cite{LiZh17,DuWeChCaAl17,HoYuWeYu17,DuXiReSuWa18,ReAb19} and  cybernetics literature~\cite{WuTaCaZh16,WeLi17,QiZhYiLi19,YiZhSt20}.


Although the FJ model is an extension of the DeGroot model, the former is significantly different from the latter. In the DeGroot model, each node has only one opinion, while  the FJ model  associates each node with two opinions: internal opinion and expressed opinion. Since its establishment, the FJ model has attracted much attention.  A sufficient condition for stability of the FJ model was obtained in~\cite{RaFrTeIs15}, and the equilibrium expressed opinion was derived in~\cite{DaGoPaSa13,BiKlOr15}. Some interpretations of the FJ model were provided in~\cite{BiKlOr15} and~\cite{GhSr14}.  And some optimization problems based on the   FJ model were also introduced, such as opinion  maximization~\cite{GiTeTs13}. Moreover, further extensions of the FJ model were suggested and studied in recent papers~\cite{JiMiFrBu15,SeGrSqRa19}. For example, some multidimensional extensions have been presented for the FJ model~\cite{PaPrTeFr16, FrPrTePa16}.

Except for the properties,  interpretations, and extension of the  FJ model, some measures for disagreement~\cite{MuMuTs18}, polarization~\cite{MaTeTs17, MuMuTs18},  conflict~\cite{ChLiDe18}, and controversy~\cite{ChLiDe18} for this popular model have also been developed and studied. These quantitative metrics provide deep insight into understanding  social phenomena. However, exact computation for  these key  measures is difficult and even impossible for large graphs, since it takes cube running time. In this paper, we give  a computationally
cheaper approach for approximating these quantities.

\section{Preliminary}
In this section, we briefly introduce some basic concepts about undirected weighted graphs, Laplacian matrix,  spanning (rooted) forest, forest matrix, FJ  opinion dynamics model, its  relevant  quantities and their connections  with the forest matrix. 

\subsection{Graph and  Laplacian Matrix}

Let $\calG= (V,E,w)$ be a connected undirected weighted simple graph with $n$ nodes and $m$ edges,  where $V=\{v_1,v_2,\cdots,v_n\}$ is the node  set,  $E=\{e_1,e_2,\cdots,e_m\}$ is the edge set, and $w : E \to \rea_{+}$ is the edge weight function, with the weight of an edge $e$ denoted by $w_e$. Let $w_{\rm max}$ and $w_{\rm min}$ denote, respectively, the maximum and minimum weight among all edges in $E$. A graph is a tree if it is connected but has no cycles. We consider a graph with only an isolated node as  a tree. A forest is a particular graph that is a disjoint union of trees. Thus, a  forest may be  connected or disconnected.  In the sequel, we interchangeably use $v_i$ and $i$ to represent node $v_i$ if incurring no confusion.

The connections  of  graph $\calG$ are encoded in its extended adjacency matrix $\AAA=(w_{ij})_{n \times n}$, with the element $w_{ij}$ at row $i$ and column $j$ representing the strength of connection   between  nodes $i$ and $j$. If nodes $i$ and $j$ are adjacent by an edge $e$ with weight $w_e$, then $w_{ij}= w_{ji}= w_e $; $w_{ij}=w_{ji}=0$ otherwise. Let $\Theta_i$ be the set of neighbours of node $i$. Then the weighted degree of a node $i$ is $d_i=\sum_{j=1}^n w_{ij}=\sum_{j\in \Theta_i} w_{ij}$. The weighted diagonal degree matrix of  $\calG$ is defined as ${\DD} = {\rm diag}(d_1, d_2, \ldots, d_n)$, and the Laplacian matrix of $\calG$ is defined to be ${\LL}={\DD}-{\AAA}$.



An alternative construction of $\LL$ is to use the incidence matrix $\BB \in \mathbb{R}^{|E| \times |V|}$, which is an $m\times n$ signed edge-node incidence matrix. The entry $b_{ev}$, $e\in E$ and $ v\in V$,  of $\BB$ is defined as follows: $b_{ev}=1$ if node $v$ is the head of edge $e$, $b_{ev}=-1$ if node $v$ is the tail of edge $e$, and $b_{ev}=0$ otherwise. Let $\ee_i$ denote the $i$-th standard basis vector. For an edge $e\in E$ with two end nodes $i$ and $j$, the row vector of $\BB$ corresponding to  $e$ can be written as $\bb_{ij}\triangleq\bb_{e}=\ee_i-\ee_j$.  Let  $\WW= {\rm diag}(w_1, w_2, \ldots, w_m)$  be an $m\times m$ diagonal matrix with the $e$-th diagonal entry being the weight of edge $w_e$. Then the Laplacian matrix $\LL$ of $\calG$ can also be represented as $\LL = \BB^\top \WW \BB$.  Moreover, $\LL$ can be written as the sum of product of block matrices as  $\LL = \sum\nolimits_{e\in E} w_e \bb_e \bb_e^\top$, which indicates that $\LL$ is a symmetric and positive semidefinite matrix.

The positive semidefiniteness of Laplacian matrix $\LL$ implies that   all its  eigenvalues are non-negative. Moreover, for a connected graph $\calG$, its  Laplacian matrix $\LL$ has a unique zero eigenvalue. Let $\textbf{1}$ denote the $n$-dimensional column vector with all entries being ones, i.e. $\textbf{1}=\left(1,1\cdots,1\right)^\top$, which is  an eigenvector of  $\LL$ associated with  eigenvalue 0. That is, $\LL \textbf{1}=\textbf{0}$, where  $\textbf{0}$ is the zero vector.  Let $\lambda_1\ge\lambda_2\ge\ldots\ge\lambda_{n-1} \ge \lambda_n=0$ be the $n$ eigenvalues of  $\LL$, and let $\uu_i$ be the orthogonal eigenvector corresponding to $\lambda_i$. Then,  $\LL$ has an eigendecomposition of form $\LL=\UU \Lambda \UU^\top=\sum_{i=1}^{n-1} \lambda_i \uu_i\uu_i^{\top}$ where $\Lambda=\textrm{diag}\left(\lambda_1,\lambda_2,..,\lambda_{n-1},0\right)$ and $\uu_i$ is the $i$-th column of matrix $\UU$.  Let $\lambda_{\max}$ and  $\lambda_{\min}$  be, respectively, the maximum and nonzero minimum  eigenvalue of  $\LL$. Then,  $\lambda_{\max}= \lambda_{1}\leq n\, w_{\rm max}$~\cite{SpSr11}, and $\lambda_{\min}=\lambda_{n-1}\geq w _{\min}/ n^2 $~\cite{LiSc18}.

\subsection{Spanning  Forests and Forest Matrix}\label{Sec:bound}

For a graph  $\mathcal{G}=(V,E,w)$, a subgraph $\mathcal{H}$ is a graph whose sets of nodes and edges are subsets of $V$ and $E$, respectively. If $\mathcal{H}$ and $\mathcal{G}$ have the same node set $V$, we call $\mathcal{H}$ a spanning subgraph of $\mathcal{G}$. A spanning forest on $\mathcal{G}$ is a spanning subgraph of $\mathcal{G}$ that is a forest. A spanning rooted forest of $\mathcal{G}$ is a spanning forest of $\mathcal{G}$, where each tree has a node marked as its root. For a subgraph $\mathcal{H}$ of graph $\mathcal{G}$, the product of the weights of all edges in $\mathcal{H}$ is referred to as the weight of  $\mathcal{H}$, denoted as  $\varepsilon(\mathcal{H})$. If $\mathcal{H}$ has no edges, its weight is set to be 1. For any nonempty set $S$ of subgraphs, we define its weight $\varepsilon(S)$ as $\varepsilon(S) =\sum_{\mathcal{H} \in S}  \varepsilon(\mathcal{H})$. If $S$ is empty, we set its weight to be zero~\cite{ChSh97,ChSh98}.

Suppose that $\Gamma$ is the set of all spanning rooted forests of graph $\calG$ and $\Gamma_{i j}$ is the set of those spanning forests of $\calG$ with nodes $v_i$ and $v_j$  in  the same tree rooted at node $v_i$. Based on the above notions associated with spanning rooted forests, we can define the forest matrix $\mathbf{\Omega}=\mathbf{\Omega}(\calG)$ of graph $\calG$~\cite{GoDrRo81,Ch08}. Let $\II$ be the identity matrix. Then  the forest matrix  is defined as  $\mathbf{\Omega}=\left(\II+\LL\right)^{-1}=(\omega_{ij})_{n \times n}$, where the entry $\omega_{ij}=\varepsilon(\Gamma_{i j})/ \varepsilon(\Gamma)$~\cite{ChSh97,ChSh98}.  For  an arbitrary pair of nodes $v_i$ and $v_j$ in graph $\calG$, $\omega_{ij}\geq 0$ with equality if and only if $\calG$ is disconnected. Moreover, $\omega_{ij}= 0$ if and only if there is no path between  $v_i$ and $v_j$~\cite{Me97}. 

If every edge in $\calG$ has unit  weight, then  $ \varepsilon(\Gamma)$ is equal to the total number of
spanning rooted forests of  $\calG$, and   $\varepsilon(\Gamma_{i j}) $ equals the number of spanning rooted forests of $\calG$, where nodes  $v_i$ and $v_j$ are in the same tree rooted at $v_i$. For example, in the  $5$-node path graph  $\calP_5$, there are exactly $55$ spanning rooted forests, among which there are  $13$   forests where $v_2$ belongs to a tree rooted at $v_1$.  Figure~\ref {SFFP5} illustrates all the 55 spanning rooted forests in  $\calP_5$, where the  13  spanning rooted forests with green background are those, for each of which $v_1$ and $v_2$ belong to the same  tree with $v_1$ being the root. According to Fig.~\ref {SFFP5}, the forest matrix for  graph  $\calP_5$ is
\begin{equation}\label{ForestMatrix}
\bm{\Omega} = \frac{1}{55}\begin{pmatrix}
34 &  13 &   5 &   2 &   1 \\
13 &  26 &  10 &   4 &   2 \\
5 &  10 &  25 &  10 &   5 \\
2 &   4 &  10 &  26 &  13 \\
1 &   2 &   5 &  13 &  34 \\
\end{pmatrix}.\notag
\end{equation}
\begin{figure}
	\begin{center}
		\includegraphics[width=0.96\linewidth]{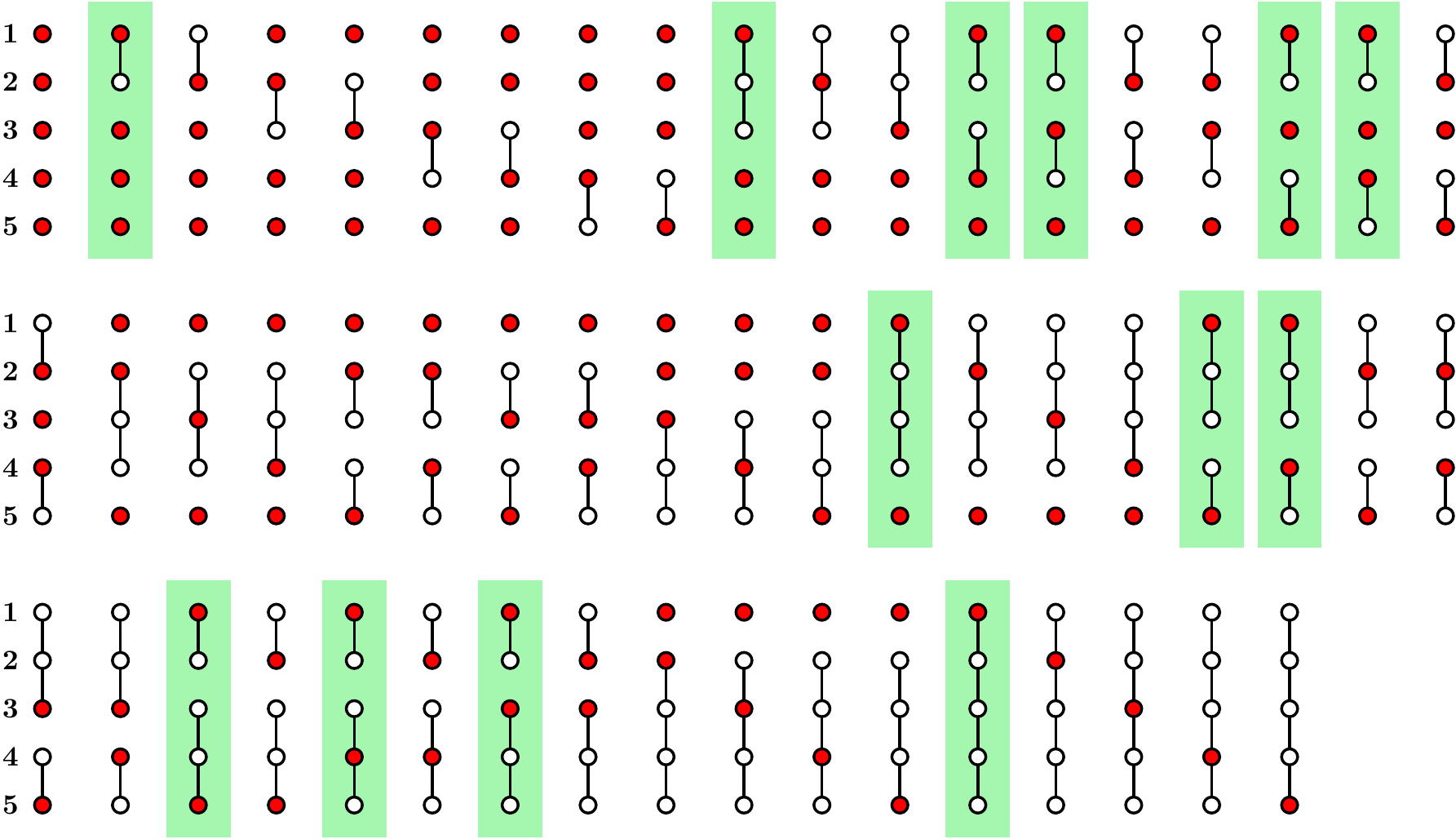}
		\caption{All the spanning rooted forests in  the path graph  $\calP_5$ with $5$ nodes. Those with green background are  forests with $v_2$ being in a tree rooted at $v_1$. The full nodes denote root nodes. }
		\label{SFFP5}
	\end{center}
\end{figure}

The forest matrix  $\mathbf{\Omega}$ is a  symmetric and positive definite matrix, the eigendecomposition of which can be written as $\mathbf{\Omega}=\UU \tilde{\Lambda} \UU^\top=\sum_{i=1}^{n} \frac{1}{\lambda_i+1} \uu_i\uu_i^{\top}$ where $\tilde{\Lambda}$ is a diagonal matrix given by \begin{equation}\nonumber
\tilde{\Lambda}=\textrm{diag}(\frac{1}{1+\lambda_1}, \cdots,\frac{1}{1+\lambda_{n-1}},\frac{1}{1+\lambda_n}),
\end{equation} 
with $\frac{1}{1+\lambda_1}\le\frac{1}{1+\lambda_2}\le\cdots\le\frac{1}{1+\lambda_n}=1.$ It has been shown that~\cite{ChSh97,ChSh98}  $\mathbf{\Omega}$ is a doubly stochastic matrix satisfying  $\mathbf{\Omega} \textbf{1}=\textbf{1}$ and $\textbf{1}^\top \mathbf{\Omega}= \textbf{1}^\top $. For any connected graph, $\omega_{ij}> 0$. Moreover, $\sum_{j=1}^{n}\omega_{ij}=1$ for $i=1,2,\ldots,n$, and $\sum_{i=1}^{n}\omega_{ij}=1$
or $j=1,2,\ldots,n$.

The forest matrix $\mathbf{\Omega}$  is  related to various  practical applications~\cite{FoPiReSa07,SeGaMaShSaFo14,JiBaZh19}.  For example,  its entry $\omega_{ij} $ can be used to  gauge the \textit{proximity} between  nodes $v_i$ and $v_j$: the less the value of $\omega_{ij} $, the ``farther''  $v_i$ from $v_j$~\cite{ChSh97}. Furthermore, since  $\mathbf{\Omega}$ is  doubly stochastic, $\omega_{ij} $ can be explained as the fraction of the connectivity of $v_i$ and $v_j$ in the total connectivity of $v_i$ (or $v_j$) with all nodes~\cite{ChSh98}.

Particularly,  forest matrix is in fact the fundamental matrix~\cite{MaTeTs17} of the FJ  opinion dynamics model~\cite{FrJo90}.  Various important  quantities  of  the FJ model can be expressed in terms of the linear combination of the entries  for forest matrix or quadric forms of  forest matrix  or its variant matrices~\cite{BiKlOr15,ChLiDe18,MuMuTs18}.  In the sequel, we will show that by using the properties of forest matrix, one can provide a fast algorithm evaluating relevant quantities for the FJ opinion dynamics model.

\section{Friedkin-Johnsen Opinion Dynamics Model and Its Relevant Quantities}

This section is denoted to brief introduction to the  FJ  model of  opinion formation, as well as the definitions and measures for  conflict, disagreement, polarization, and controversy, relying on this popular model. Particularly, we give an explanation and some properties of equilibrium  expressed opinions of the FJ model, using the  forest matrix.

\subsection{Friedkin-Johnsen  Model}

As one of the first opinion dynamics models, the FJ model~\cite{FrJo90}  is an extension of the DeGroot's opinion model~\cite{De74}.  In the DeGroot model, every node has only one opinion that is updated as the weighted average of its neighbors. Different from the DeGroot model,  in the FJ model, each node $i \in V$ has two different kinds of opinions: one is the internal (or innate) opinion $s_i$, the other is the expressed opinion $z_i$. The internal  opinion $s_i$  is assumed to remain  constant,  private to node $i$, while the  expressed opinion $z_i$ evolves as a weighted average of  its corresponding internal opinion $s_i$ and the expressed opinions of $i$'s neighbors.   More precisely, the updating rule of  $z_i$ is
\begin{equation}\label{FJmodel}
z_{i}=\frac{ s_{i}+\sum_{j \in \Theta_i} w_{i j} z_{j}}{1+\sum_{j \in \Theta_i} w_{i j}}.
\end{equation}

Note that in the above updating process, we make a common assumption in the literature that the weight of internal  opinion is unit. On the other hand, as popular choice in the literature, we assume that for all $i \in V$, its internal  opinion $s_i$ is in the interval $[0,1]$. Let $\vs=(s_1,s_2,\ldots,s_n)^\top$. Then the expressed opinions updated by iterative process~\eqref{FJmodel} converge to a unique equilibrium opinion vector. Let $\zz=(z_1, z_2,\ldots,z_n)^\top$,  with the value $z_i$ being the  expressed opinion of node $i$ at equilibrium. It was shown~\cite{BiKlOr15} that the equilibrium  expressed opinion vector $\zz$  is the solution to a linear system of equations:
\begin{equation}\label{FJmodel02}
\zz =(\II+\LL)^{-1} \vs .
\end{equation}

Equation~\eqref{FJmodel02} shows that the equilibrium  expressed opinion for every node is determined by the forest matrix $\mathbf{\Omega}=\left(\II+\LL\right)^{-1}$, with  $z_i$ given by $z_i=\sum^n_{j=1}  \omega_{ij}s_j $, for each $i \in V$. Concretely,  for all $i=1,2\ldots,n$,   $z_i$ is a weighted average of internal  opinions of all nodes, with the weight of internal  opinion $s_j$ being  $ \omega_{ij}$, where  $j=1,2\ldots,n$.  Considering that  $\mathbf{\Omega}$  is a doubly stochastic and $s_i \in [0,1]$ for all $i=1,2\ldots,n$, we have that for every $i$, $z_i  \in [0,1]$. Moreover, $\sum^n_{i=1}  s_i=\sum^n_{i=1}z_i $, which means that the  total expressed opinion is equal to the total internal opinion, although the equilibrium  expressed opinion for a single node may be different from its  internal  opinion. This conservation law is independent  of the network structure. In this sense, we provide a novel interpretation and some properties  of  equilibrium  expressed opinion vector $\zz$ according to the  forest matrix. Our interpretation is different from previous ones, which are in terms of game theory~\cite{BiKlOr15} and electrical networks~\cite{GhSr14}, respectively.

\subsection{Measures for Conflict, Disagreement, Polarization, and Controversy}

In the FJ model, the opinions of nodes often do not reach consensus, leading to conflict, disagreement, polarization, and controversy, which  are common phenomena in social networks  and have been the subject of many recent works. Below we survey  some quantitative measures of these phenomena based on the FJ  opinion formation model.

As known to us all, in the FJ model,  individuals  differ in their internal opinions and expressed
opinions. The extent of this difference can be measured by internal conflict defined as follows~\cite{ChLiDe18}.
\begin{definition}
	For a graph  $\calG= (V,E,w)$, its internal conflict $C_{\rm I}(\calG)$ is the sum of squares of
	the differences between internal and expressed opinions over all nodes:
	\begin{equation}\label{eq:dfn_Ci}
	C_{\rm I}(\calG) =  \sum_{i\in V}\left(z_{i}-s_{i}\right)^{2}.
	\end{equation}
\end{definition}

\begin{definition}~\cite{MuMuTs18, DaGoLe13}
	For a graph $\calG= (V,E,w)$, its  disagreement $D(\calG)$  is defined by
	\begin{equation}\label{eq:dfn_disagree}
	D(\calG) =   \sum\limits_{(i,j) \in E} w_{ij} (z_i-z_j)^2.
	\end{equation}
\end{definition}

Note that the  disagreement $D(\calG)$ is called external conflict of graph $\calG$ in~\cite{ChLiDe18}.  We also note that the FJ  model has been used to understand  the price of anarchy in society when individuals selfishly update their opinions with an aim  to minimize the stress they experience~\cite{BiKlOr15}.   The stress of a node $i$ is defined as $(z_i - s_i)^2 + \sum_{j \in \Theta_i} w_{ij}(z_i-z_j)^2$, while the sum of the stress for all nodes is $C_{\rm I}(\calG)+D(\calG)$, which is exactly the sum of internal conflict and external conflict defined above.

If the equilibrium  expressed opinions have an increased divergence, we say that opinion formation dynamics are polarizing. Intuitively, polarization should measure how  equilibrium  expressed opinions  deviate from their average. There are many ways to quantify polarization. We here choose the metric proposed  in~\cite{MuMuTs18} to measure polarization.

\begin{definition}
	For a graph  $\calG= (V,E,w)$,  let $ \bar{z} $ be the mean-centered equilibrium vector given by
	$ \bar{\zz} =z  - \frac{\zz{ \top}  \textbf{1}}{n} \textbf{1}$. Then the polarization $P(\calG)$ is defined to be:
	\begin{equation}\label{eq:dfn_contr}
	P(\calG) =  \sum\limits_{i \in V}\bar{z} _i^2 =\bar{\zz}^{\top} \bar{\zz}.
	\end{equation}
\end{definition}	
In addition to $P(\calG)$,  the  polarization can also  be measured by the controversy that quantifies
how much the expressed opinion varies across the individuals in the whole graph $\calG$.
\begin{definition}
	For a graph  $\calG= (V,E,w)$, the controversy $C(\calG)$ is the sum of the squares of the
	equilibrium expressed opinions:
	\begin{equation}\label{eq:dfn_contrCG}
	C(\calG) =  \sum\limits_{i \in V}  z _i^2 =\zz^\top \zz.
	\end{equation}
\end{definition}
In~\cite{MaTeTs17}, the quantity $C(\calG)/n$ is introduced as the polarization index.

There can be a tradeoff between disagreement and controversy~\cite{MuMuTs18}. The sum of the disagreement and controversy is called disagreement-controversy index, which is named polarization-disagreement index in~\cite{MuMuTs18}.
\begin{definition}
	For a graph  $\calG= (V,E,w)$, the  disagreement-controversy index $I_{\rm dc}(\calG)$ is the sum of the  disagreement  $D(\calG)$  and controversy  $C(\calG)$ :
	\begin{equation}\label{eq:dfn_DisCon}
	I_{\rm dc}(\calG) = D(\calG)+ C(\calG).
	\end{equation}
\end{definition}
It is easy to verify~\cite{ChLiDe18,MuMuTs18} that the  disagreement-controversy index $I_{\rm dc}(\calG)$ is equal to the inner product between
internal  opinion vector $ \vs$ and expressed opinion vector  $\zz$, that is,  $I_{\rm dc}(\calG)=\sum_{i=1}^n s_{i} z_{i}$.

Let $\bar{\vs} = \vs - \frac{\bm{1}^{\top} \vs}{n} \bm{1}$.  The above-mentioned quantities can be expressed in matrix-vector notation as stated in  Proposition~\ref{Prop}~\cite{ChLiDe18,MuMuTs18}.
\begin{proposition}\label{Prop}
	For a graph  $\calG= (V,E,w)$, the internal conflict $C_{\rm I}(\calG)$,  disagreement  $D(\calG)$,  polarization $P(\calG)$, controversy  $C(\calG)$, and disagreement-controversy index $I_{\rm dc}(\calG)$ can be  conveniently expressed in terms of quadratic forms as:
	\begin{equation}\label{eq:dfn_CiA}
	C_{\rm I}(\calG) =\zz^{\top} \LL^{2} \zz=\vs^{\top}(\II+\LL)^{-1} \LL^{2}(\II+\LL)^{-1} \vs,
	\end{equation}
	\begin{equation}\label{eq:dfn_disagreeA}
	D(\calG) = \zz^{\top} \LL \zz=\vs^{\top}(\II+\LL)^{-1} \LL(\II+\LL)^{-1} \vs,
	\end{equation}
	\begin{equation}\label{eq:dfn_contrA}
	P(\calG)  =\bar{\zz}^{\top} \bar{\zz}=\bar{\vs}^{\top}(\II+\LL)^{-1}(\II+\LL)^{-1} \bar{\vs},
	\end{equation}
	\begin{equation}\label{eq:dfn_contrCGA}
	C(\calG) =\zz^{\top} \zz=\vs^{\top}(\II+\LL)^{-1} (\II+\LL)^{-1}\vs,
	\end{equation}
	\begin{equation}\label{eq:dfn_DisConA}
	I_{\rm dc}(\calG) = \vs^{\top}\zz=\vs^{\top }(\II+\LL)^{-1} \vs.
	\end{equation}
\end{proposition}
Notice that  matrices $\LL$ and $(\II+\LL)^{-1}$  have identical  eigenspaces, implying  that they commute, that is, $\LL(\II+\LL)^{-1}=(\II+\LL)^{-1}\LL$. Thus, we have
$$C_{\rm I}(\calG) = \bar{\vs}^{\top}(\II+\LL)^{-1} \LL^{2}(\II+\LL)^{-1}  \bar{\vs},$$
$$D(\calG) = \bar{\vs}^{\top}(\II+\LL)^{-1} \LL(\II+\LL)^{-1} \bar{\vs},$$
due to $\LL \textbf{1}=\textbf{0}$.

After expressing the quantities concerned, in what follows we will provide a fast algorithm evaluating these quantities.

\section{Fast Approximation Algorithm for Conflict, Disagreement, Polarization, and Controversy}

As shown in Proposition~\ref{Prop}, for internal conflict $C_{\rm I}(\calG)$,  disagreement  $D(\calG)$,  polarization $P(\calG)$, controversy  $C(\calG)$, and disagreement-controversy index $I_{\rm dc}(\calG)$,  exactly computing  them  needs to invert matrix $\LL+\II$, which takes $O(n^3)$  time. This is computationally impractical for large graphs.

In this section, we  develop  a fast  algorithm for approximately evaluating those interesting quantities in nearly linear time with respect of the number of edges in $\calG$. To achieve this goal, we first reduce the problem for  evaluating the above quantities  to computing the $\ell_2$  norm of different  vectors. Then, we  estimate the $\ell_2$ norm by applying linear system solvers~\cite{CoKyMiPaPeRaXu14,KySa16} in order to significantly  reduce the computational complexity.

According to  Proposition~\ref{Prop}, we can explicitly represent  the concerned quantities in $\ell_2$  norm of vectors as stated in Lemma~\ref{PropA}.
\begin{lemma}\label{PropA}
	For a graph  $\calG= (V,E,w)$, the internal conflict $C_{\rm I}(\calG)$,  disagreement  $D(\calG)$,  polarization $P(\calG)$, controversy  $C(\calG)$, and disagreement-controversy index $I_{\rm dc}(\calG)$ can be   expressed, respectively, in terms of  $\ell_2$  norm as:
	\begin{equation}\label{eq:dfn_CiA}
	C_{\rm I}(\calG) =\zz^{\top} \LL^{2} \zz=\| \LL (\II+\LL)^{-1} \vs\| ^2,
	\end{equation}
	\begin{equation}\label{eq:dfn_disagreeA}
	D(\calG) = \zz^{\top} \LL \zz=\| \WW^{1/2} \BB (\II+\LL)^{-1} \vs\| ^2,
	\end{equation}
	\begin{equation}\label{eq:dfn_contrA}
	P(\calG)  =\bar{\zz}^{\top} \bar{\zz}=\| (\II+\LL)^{-1} \bar{\vs}\| ^2,
	\end{equation}
	\begin{equation}\label{eq:dfn_contrCGA}
	C(\calG) =\zz^{\top} \zz=\| (\II+\LL)^{-1} \vs\| ^2,
	\end{equation}
	\begin{equation}\label{eq:dfn_DisConA}
	I_{\rm dc}(\calG) = \| \WW^{1/2} \BB (\II+\LL)^{-1} \vs\| ^2+\| (\II+\LL)^{-1} \vs\| ^2,
	\end{equation}
	where $\WW^{1/2}$ is a diagonal matrix defined as $\WW^{1/2}= {\rm diag}(\sqrt{w_1}, \sqrt{w_2}, \\ \sqrt{w_3}, \ldots, \sqrt{w_m})$.
\end{lemma}

Having reduced the computation of the relevant quantities to evaluating $\ell_2$ norms of  some vectors  in $\mathbb{R}^n$ or $\mathbb{R}^m$, we continue to compute the $\ell_2$ norms. However, directly calculating the $\ell_2$ norms  does not help to reduce the computational cost, since it still requires  inverting   matrix $\II+\LL$. In order to reduce  computational time, we resort to the efficient linear system solvers~\cite{CoKyMiPaPeRaXu14}, which  avoids the inverse operation by solving a system of equations~\cite{KySa16}.

\begin{lemma}\label{ST}
	There is a nearly linear time solver $\yy=\textsc{Solve}\left(\TT,\xx,\delta\right)$, which takes   an $n \times n$ positive semi-definite matrix $\TT$ with $m$ non-zero entries, a column vector $\xx$, and an accuracy parameter $\delta$, and returns   a column vector $\yy$ satisfying
	$\|\yy-\TT^{\dagger}\xx\|_{\TT}\le\delta\|\TT^{\dagger}\xx\|_{\TT}$, where $\|\boldsymbol{\mathit{v}}\|_{\TT}=\sqrt{\boldsymbol{\mathit{v}}^\top \TT \boldsymbol{\mathit{v}}}$ and $\TT^{\dagger}$ is the pseudo-inverse of $\TT$. The expected time for performing this solver is $O\left(m\log^{3}{n}\log\left(\frac{1}{\delta}\right)\right)$.
\end{lemma}

Lemma~\ref{ST} can significantly reduce  the computational time for evaluating those quantities in the form of $(\II+\LL)^{-1}\xx$ with an ideal approximation guarantee. For example, as will be shown below, $\pp = \textsc{Solve}(\II + \LL, \vs, \delta)$ is a good approximation of the expressed opinion vector  $\zz$.

We next  use  Lemma~\ref{ST} to obtain  approximations for the quantities concerned. Prior to this, we introduce some notations and their properties. Let $a\geq 0$ and $b\geq 0$ be two nonnegative scalars. 	We say $a$ is an $\eps$-approximation ($0\leq \eps \leq1/2$) of $b$ if $(1-\eps) a \leq b \leq (1+\eps) a$, denoted by $a \approx_{\eps} b$.
The $\eps$-approximation has the following basic property: For nonnegative scalars $a$, $b$, $c$, and $d $, if $a \approx_\eps b$ and $c \approx_\eps d$, then $a + c \approx_\eps b + d$. For two matrices $\XX$ and $\YY$, we write $\XX \preceq \YY$ if $\YY - \XX$ is positive semidefinite, that is, $\xx^\top \XX \xx \leq \xx^\top \YY \xx$ holds  for every real vector $\xx$. Then, we have $\II \preceq \II + \LL$,  $\LL \preceq \II + \LL$, $ \II + \LL \preceq (nw_{\rm max}+1) \II$, and $ \frac{1}{n\,w_{\rm max}}\LL \preceq \II$. These relations are useful to prove the following lemmas, which we will apply to obtain $\eps$-approximation for  the quantities we care about.

\begin{lemma}\label{lm1}
	Given an undirected weighted graph $\calG=(V,E,w)$ with  each edge weight in the interval $[w_{\rm min}, w_{\rm max}]$,  the Laplacian matrix $\LL$, and a  parameter $\epsilon \in (0, \frac{1}{2})$, let $\xx$ be an arbitrary vector, and let $\yy = \textsc{Solve}\kh{\II + \LL, \xx, \delta}$, where
	\begin{align*}
	\delta \leq \frac{\epsilon}{3\sqrt{nw_{\rm max}+1}}.
	\end{align*}
	Then, the following relation holds:
	\begin{align*}
	(1-\epsilon ) \mo{(\II+\LL)^{-1}\xx}^2 \leq \mo{\yy}^2 \leq (1+\epsilon ) \mo{(\II+\LL)^{-1}\xx}^2.
	\end{align*}
\end{lemma}
\begin{proof}
	According to Lemma~\ref{ST}, we have
	\begin{align*}
	\mo{\yy - (\II+\LL)^{-1}\xx}^2_{\II+\LL} \leq \delta^2 \mo{(\II+\LL)^{-1}\xx}_{\II+\LL}^2.
	\end{align*}	
	The term on the left-hand side (lhs) can be bounded as
	\begin{align*}
	\mo{\yy - (\II+\LL)^{-1}\xx}_{\II+\LL}^2 	\geq & \mo{\yy - (\II+\LL)^{-1}\xx}^2 \\
	\geq & \len{\mo{\yy} - \mo{(\II+\LL)^{-1}\xx}}^2,
	\end{align*}
	while the term on the right-hand side (rhs) can be bounded as
	\begin{align*}
	\mo{(\II+\LL)^{-1}\xx}^2_{\II+\LL} \leq (nw_{\rm max}+1)\mo{(\II+\LL)^{-1}\xx}^2.
	\end{align*}
	Combining the above-obtained relations, we have
	\begin{align*}
	&\quad \len{\mo{\yy} - \mo{(\II+\LL)^{-1}\xx}}^2 \\
&\leq \delta^2 (nw_{\rm max}+1)\mo{(\II+\LL)^{-1}\xx}^2,
	\end{align*}
	which implies
	\begin{align*}
	\frac{\len{\mo{\yy} - \mo{(\II+\LL)^{-1}\xx}}}{\mo{(\II+\LL)^{-1}\xx}} \leq \sqrt{\delta^2 (nw_{\rm max}+1)} \leq \frac{\epsilon}{3}
	\end{align*}
	and
	\begin{align*}
	(1-\frac{\epsilon}{3} )^2 \mo{(\II+\LL)^{-1}\xx}^2 & \leq \mo{\yy}^2\\ & \leq (1+\frac{\epsilon}{3} )^2 \mo{(\II+\LL)^{-1}\xx}^2.
	\end{align*}
	Considering $0 < \epsilon < \frac{1}{2}$, we get
	\begin{align*}
	(1-\epsilon ) \mo{(\II+\LL)^{-1}\xx}^2 &\leq \mo{\yy}^2\\ & \leq (1+\epsilon ) \mo{(\II+\LL)^{-1}\xx}^2,
	\end{align*}
	which completes the proof.
\end{proof}

\begin{lemma}\label{lm2a}
	Given an undirected weighted graph $\calG=(V,E,w)$ with  each edge weight in the interval $[w_{\rm min}, w_{\rm max}]$,  the incident matrix $\BB$,  diagonal edge weight matrix $\WW$, Laplacian matrix $\LL$, and a  parameter $\epsilon \in (0, \frac{1}{2})$, let $\vs = (s_1, s_2, \ldots, s_n)^{\top}$ be the internal opinion vector with each $s_i \in [0,1]$ for $i=1, 2, \ldots, n$,  $\bar{\vs} = \vs - \frac{\bm{1}^{\top} \vs}{n} \bm{1}= (\bar{s}_1, \bar{s}_2, \ldots, \bar{s}_n)^{\top}$ be the mean-centered internal opinion vector, and  let $\qq = \textsc{Solve}\kh{\II + \LL, \bar{\vs}, \delta}$, where
	\begin{align*}
	\delta \leq \frac{\epsilon \mo{\bar{\vs}} }{3n (nw_{\rm max}+1)} \sqrt{\frac{ w_{\rm min} }{n(nw_{\rm max}+1)}}.
	\end{align*}
	Then, the following relation holds:
	\begin{align*}
	&(1-\epsilon ) \mo{\WW^{1/2}\BB(\II+\LL)^{-1}\bar{\vs}}^2 \\
	\leq & \mo{\WW^{1/2}\BB\qq}^2\leq (1+\epsilon ) \mo{\WW^{1/2}\BB(\II+\LL)^{-1}\bar{\vs}}^2.
	\end{align*}
\end{lemma}
\begin{proof}
	By Lemma~\ref{ST}, we have
	\begin{align*}
	\mo{\qq - (\II+\LL)^{-1}\bar{\vs}}^2_{\II+\LL} \leq \delta^2\mo{(\II+\LL)^{-1} \bar{\vs}}^2_{\II+\LL}.
	\end{align*}	
	The lhs can be bounded as
	\begin{align*}
	& \mo{\qq - (\II+\LL)^{-1}\bar {\vs}}^2_{\II+\LL} \\
	\geq & \mo{\qq - (\II+\LL)^{-1}\bar{\vs}}^2_{\LL}\\ = & \mo{\WW^{1/2}\BB\qq - \WW^{1/2}\BB(\II+\LL)^{-1}\bar{\vs}}^2 \\
	\geq & \len{\mo{\WW^{1/2}\BB\qq} - \mo{\WW^{1/2}\BB(\II+\LL)^{-1}\bar{\vs}}}^2,
	\end{align*}
	while  rhs is bounded as
	\begin{align*}
	\mo{(\II+\LL)^{-1}\bar{\vs}}_{\II+\LL}^2 &\leq (nw_{\rm max}+1)\mo{(\II+\LL)^{-1}\bar{\vs}}^2\\
	&\leq  n(n w_{\rm max}+1),
	\end{align*}
	where $\bar{\vs}_i\leq 1$, $i=1, 2, \ldots, n$, is used.
	Combining the above-obtained results gives
	\begin{align*}
	&\quad \len{\mo{\WW^{1/2}\BB\qq} - \mo{\WW^{1/2}\BB(\II+\LL)^{-1}\bar{\vs}}}^2\\ &\leq \delta^2 n(n w_{\rm max}+1).
	\end{align*}
	On the other hand, since $\bar{\vs}^{\top}\bm{1}=0$ due to $\bar{\vs} = \vs - \frac{\bm{1}^{\top} \vs}{n} \bm{1}$,
	\begin{align*}
	\mo{\WW^{1/2}\BB(\II+\LL)^{-1}\bar{\vs}}^2
	\geq  \frac{w_{\rm min}}{n^2(n w_{\rm max}+1)^2}\mo{\bar{\vs}}^2.
	\end{align*}
	Thus, one has
	\begin{align*}
	&\frac{\len{\mo{\WW^{1/2}\BB\qq} - \mo{\WW^{1/2}\BB(\II+\LL)^{-1}\bar{\vs}}}}{\mo{\WW^{1/2}\BB(\II+\LL)^{-1}\bar{\vs}}} \\
	\leq & \sqrt{\frac{\delta^2 n^3(nw_{\rm max}+1)^3}{w_{\rm min} \mo{\bar{\vs}}^2} } \leq \frac{\epsilon}{3}.
	\end{align*}
	In other words,
	\begin{align*}
	&(1-\frac{\epsilon}{3} )^2 \mo{\WW^{1/2}\BB(\II+\LL)^{-1}\bar{\vs}}^2 \\
	\leq& \mo{\WW^{1/2}\BB\qq}^2 \leq (1+\frac{\epsilon}{3} )^2 \mo{\WW^{1/2}\BB(\II+\LL)^{-1}\bar{\vs}}^2.
	\end{align*}
	Because $0 < \epsilon < \frac{1}{2}$, we obtain
	\begin{align*}
	&(1-\epsilon ) \mo{\WW^{1/2}\BB(\II+\LL)^{-1}\bar{\vs}}^2 \\
	\leq& \mo{\WW^{1/2}\BB\qq}^2 \leq (1+\epsilon ) \mo{\WW^{1/2}\BB(\II+\LL)^{-1}\bar{\vs}}^2,
	\end{align*}
	which completes the proof.
\end{proof}

\begin{lemma}\label{lm3}
	Given an undirected weighted graph $\calG=(V,E,w)$ with  each edge weight in the interval $[w_{\rm min}, w_{\rm max}]$,  Laplacian matrix $\LL$, and a  parameter $\epsilon \in (0, \frac{1}{2})$, let $\vs = (s_1, s_2, \ldots, s_n)^{\top}$ be the internal opinion vector with each $s_i \in [0,1]$ for $i=1, 2, \ldots, n$,  $\bar{\vs} = \vs - \frac{\bm{1}^{\top} \vs}{n} \bm{1}= (\bar{s}_1, \bar{s}_2, \ldots, \bar{s}_n)^{\top}$ be the mean-centered internal opinion vector,  $\qq = \textsc{Solve}\kh{\II + \LL, \bar{\vs}, \delta}$, where
	\begin{align*}
	\delta \leq \frac{\epsilon w_{\rm min}\mo{\bar{\vs}}}{3 w_{\rm max} n^3 (nw_{\rm max}+1) \sqrt{n}}.
	\end{align*}
	Then, the following relation holds:
	\begin{align*}
	(1-\epsilon ) \mo{\LL(\II+\LL)^{-1}\bar{\vs}}^2 &\leq \mo{\LL\qq}^2\\
 & \leq (1+\epsilon ) \mo{\LL(\II+\LL)^{-1}\bar{\vs}}^2.
	\end{align*}
\end{lemma}
\begin{proof}
	Making use of Lemma~\ref{ST}, we obtain
	\begin{align*}
	\mo{\qq - (\II+\LL)^{-1}\bar{\vs}}^2_{\II+\LL} \leq \delta^2 \mo{(\II+\LL)^{-1}\bar{\vs}}^2_{\II+\LL}.
	\end{align*}	
	The term on the lhs can be bounded as
	\begin{align*}
	& \mo{\qq - (\II+\LL)^{-1}\bar{\vs}}^2_{\II+\LL} \\
	\geq & \mo{\qq - (\II+\LL)^{-1}\bar{\vs}}^2 \geq \frac{1}{(nw_{\rm max})^{2}}\mo{\LL\qq -\LL(\II+\LL)^{-1}\bar{\vs}}^2 \\
	\geq & \frac{1}{(nw_{\rm max})^{2}} \len{\mo{\LL\qq} - \mo{\LL(\II+\LL)^{-1}\bar{\vs}}}^2.
	\end{align*}
	According to the proof of  Lemma~\ref{lm2a},
	\begin{align*}
	\mo{(\II+\LL)^{-1}\bar{\vs}}_{\II+\LL}^2 	\leq & n(n w_{\rm max}+1).
	\end{align*}
	Thus, we have
	\begin{align*}
	\len{\mo{\LL\qq} - \mo{\LL(\II+\LL)^{-1}\bar{\vs}}}^2 \leq \delta^2 n^3w^2_{\rm max} (n w_{\rm max}+1).
	\end{align*}
	On the other hand,
	\begin{align*}
	\mo{\LL(\II+\LL)^{-1}\bar{\vs}}^2
	\geq  \frac{w^2_{\rm min}}{n^4(n w_{\rm max}+1)^2}\mo{\bar{\vs}}^2.
	\end{align*}
	Combining the above  relations leads to
	\begin{align*}
	\frac{\len{\mo{\LL\qq} - \mo{\LL(\II+\LL)^{-1}\bar{\vs}}}}{\mo{\LL(\II+\LL)^{-1}\overline{\vs}}}
	\leq  \sqrt{\frac{\delta^2 n^7 w_{\rm max}^2(n  w_{\rm max}+1)^2}{w^2_{\rm min} \mo{\bar{\vs}}^2} } \leq \frac{\epsilon}{3},
	\end{align*}
	which can be recast as
	\begin{align*}
(1-\frac{\epsilon}{3} )^2 \mo{\LL(\II+\LL)^{-1}\bar{\vs}}^2& \leq \mo{\LL\qq}^2 \\
	&\leq(1+\frac{\epsilon}{3} )^2 \mo{\LL(\II+\LL)^{-1}\bar{\vs}}^2.
	\end{align*}
	Using the condition $0 < \epsilon < \frac{1}{2}$, one obtains
	\begin{align*}
(1-\epsilon ) \mo{\LL(\II+\LL)^{-1}\bar{\vs}}^2& \leq \mo{\LL\tilde{\zz}}^2 \\
	&\leq(1+\epsilon ) \mo{\LL(\II+\LL)^{-1}\bar{\vs}}^2,
	\end{align*}
	which completes the proof.
\end{proof}

Based on the Lemmas~\ref{ST},~\ref{lm1},~\ref{lm2a}, and~\ref{lm3}, we propose a fast and efficient  algorithm \textsc{Approxim} to approximate the  internal conflict $C_{\rm I}(\calG)$,  disagreement  $D(\calG)$,  polarization $P(\calG)$, controversy  $C(\calG)$, and disagreement-controversy index $I_{\rm dc}(\calG)$ for any undirected weighted graph $\calG$. In  Algorithm~\ref{alg:VC}, we present the pseudocode of \textsc{Approxim}, where $\delta$ is less than or equal to the parameters $\delta$'s in Lemmas~\ref{lm1},~\ref{lm2a}, and~\ref{lm3}.

\begin{algorithm}
	\caption{$\textsc{Approxim}\kh{\calG, \vs, \epsilon}$}
	\label{alg:VC}
	\Input{$\calG$: a graph with edge weight in $[ w_{\rm min},w_{\rm max} ]$ 	\\	$\vs$: initial opinion vector \\$\epsilon$: the error parameter in $ (0, \frac{1}{2})$ \\
	}
	\Output{$\{ \tilde{C}_{\rm I}(\calG), \tilde{D}(\calG),  \tilde{P}(\calG),  \tilde{C}(\calG), \tilde{I}_{\rm dc}(\calG)\}$
	}
	$\delta=\frac{\epsilon w_{\rm min}\mo{\bar{\vs}}}{3 w_{\rm max} n^3 (nw_{\rm max}+1) \sqrt{n}}$ \;
	$\bar{\vs} = \vs - \frac{\bm{1}^{\top} \vs}{n} \bm{1}$ \;
	$\tilde{\zz} = \textsc{Solve}(\II + \LL, \vs, \delta)$ \;
	$ \qq  = \textsc{Solve}(\II + \LL, \bar{\vs}, \delta)$ \;
	$\tilde{C}_{\rm I}(\calG) = \mo{\LL  \tilde{\zz} }^2$ \;
	$\tilde{D}(\calG) = \mo{\WW^{1/2}\BB \qq }^2$ \;
	$\tilde{P}(\calG)= \mo{\qq }^2$ \;
	$ \tilde{C}(\calG) = \mo{\tilde{\zz} }^2$ \;
	$ \tilde{I}_{\rm dc}(\calG) = \tilde{D}(\calG) + \tilde{C}(\calG)$ \;
	\textbf{return} $\{ \tilde{C}_{\rm I}(\calG), \tilde{D}(\calG),  \tilde{P}(\calG),  \tilde{C}(\calG), \tilde{I}_{\rm dc}(\calG) \}$ \;
\end{algorithm}

Our approximation algorithm \textsc{Approxim} is both accurate and efficient, as   summarized   in Theorem~\ref{ThmV}.
\begin{theorem}\label{ThmV}
	Given an undirected weighted graph $\calG$ with  $n$ nodes and $m$ edges,   an error   parameter $\epsilon \in (0, \frac{1}{2})$, and the internal opinion vector $\vs$,  the algorithm $\textsc{Approxim}\kh{\calG, \vs, \epsilon}$ runs in expected time  $O\left(m\log^{4}{n}\log\left(\frac{r}{\epsilon}\right)\right)$ where $r=\frac{w_{\rm max}}{w_{\rm min}}$, and returns the $\eps$-approximation $\tilde{C}_{\rm I}(\calG)$, $\tilde{D}(\calG)$,  $\tilde{P}(\calG)$,  $\tilde{C}(\calG)$, $\tilde{I}_{\rm dc}(\calG)$ for the  internal conflict $C_{\rm I}(\calG)$,  disagreement  $D(\calG)$,  polarization $P(\calG)$, controversy  $C(\calG)$, and disagreement-controversy index $I_{\rm dc}(\calG)$, satisfying
	$\tilde{C}_{\rm I}(\calG) \approx_\eps C_{\rm I}(\calG)$,
	$\tilde{D}(\calG) \approx_\eps D (\calG)$,
	$\tilde{P}(\calG) \approx_\eps P (\calG)$,
	$\tilde{C}(\calG) \approx_\eps C (\calG)$, and
	$\tilde{I}_{\rm dc}(\calG) \approx_\eps I_{\rm dc}(\calG)$.
\end{theorem}

\begin{table*}
	\fontsize{7.5}{8.0}\selectfont
	\centering
	\caption{Statistics of real-world networks used in our experiments and comparison of running time (seconds, $s$) between \textsc{Exact} and \textsc{Approxim} for three internal distributions (uniform distribution, exponential distribution, and power-law distribution) with error parameter $\epsilon=10^{-6}$. }\label{T1}
	\begin{tabular}{lrrlrrlrrlrr}
		\toprule
		&&&&&\multicolumn{6}{c}{Running time ($s$) for algorithms \textsc{Exact} and \textsc{Approxim}} \\
		\cmidrule{5-12}
		\multirow{2}*{Network} & \multirow{2}*{$n'$} & \multirow{2}*{$m'$} &&\multicolumn{2}{c}{Uniform distribution}
		&& \multicolumn{2}{c}{Exponential distribution}
		&& \multicolumn{2}{c}{Power-law distribution} \\
		\cmidrule{5-6}
		\cmidrule{8-9}
		\cmidrule{11-12}
		& & &  & \textsc{Exact}& \textsc{Approxim}&& \textsc{Exact} & \textsc{Approxim}&& \textsc{Exact}& \textsc{Approxim} \\
		\midrule
		Erd\"{o}s992 & 4,991 & 7,428 & & 2.24 & 2.66 & & 2.29 & 2.65 & & 2.51 & 2.72\\
		Advogato & 5,054 & 43,015 & & 2.75 & 2.80 & & 2.43 & 2.71 & & 2.51 & 2.77\\
		PagesGovernment & 7,057 & 89,429 & & 7.52 & 2.70 & & 9.72 & 2.75 & & 7.82 & 2.59\\
		WikiElec & 7,066 & 100,727 & & 7.59 & 2.69 & & 6.68 & 2.69 & & 7.34 & 2.63\\
		HepPh & 11,204 & 117,619 & & 28.14 & 2.74 & & 27.67 & 2.61 & & 27.58 & 2.62\\
		Anybeat & 12,645 & 49,132 & & 40.59 & 2.65 & & 40.87 & 2.74 & & 40.21 & 2.73\\
		PagesCompany & 14,113 & 52,126 & & 56.02 & 2.78 & & 55.61 & 2.96 & & 55.67 & 2.75\\
		AstroPh & 17,903 & 196,972 & & 115.50 & 2.91 & & 117.64 & 2.87 & & 117.59 & 2.82\\
		CondMat & 21,363 & 91,286 & & 204.15 & 2.83 & & 215.57 & 3.04 & & 200.75 & 2.59\\
		Gplus & 23,613 & 39,182 & & 279.41 & 2.82 & & 281.00 & 2.96 & & 274.52 & 2.75\\
		GemsecRO & 41,773 & 125,826 & & 1585.75 & 3.52 & & 1621.83 & 3.17 & & 1587.00 & 3.27\\
		GemsecHU & 47,538 & 222,887 & & 2410.30 & 6.52 & & 2395.84 & 7.24 & & 2430.26 & 6.95\\
		PagesArtist & 50,515 & 819,090 & & 2895.73 & 10.06 & & 2976.45 & 7.50 & & 3017.09 & 8.07\\
		Brightkite & 56,739 & 212,945 & & 4203.67 & 8.66 & & 4312.25 & 6.74 & & 4072.74 & 8.53\\
		Themarker & 69,317 & 1,644,794 & & --- & 5.02 & & --- & 5.04 & & --- & 5.09\\
		Slashdot & 70,068 & 358,647 & & --- & 3.31 & & --- & 3.36 & & --- & 3.34\\
		BlogCatalog & 88,784 & 2,093,195 & & --- & 5.54 & & --- & 5.50 & & --- & 5.53\\
		WikiTalk & 92,117 & 360,767 & & --- & 3.21 & & --- & 3.28 & & --- & 3.29\\
		Buzznet & 101,163 & 2,763,066 & & --- & 6.39 & & --- & 6.33 & & --- & 6.31\\
		LiveMocha & 104,103 & 2,193,083 & & --- & 6.46 & & --- & 6.44 & & --- & 6.46\\
		Douban & 154,908 & 327,162 & & --- & 3.30 & & --- & 3.35 & & --- & 3.37\\
		Gowalla & 196,591 & 950,327 & & --- & 4.43 & & --- & 4.30 & & --- & 4.29\\
		Academia & 200,167 & 1,022,440 & & --- & 4.51 & & --- & 4.44 & & --- & 4.49\\
		GooglePlus & 201,949 & 1,133,956 & & --- & 4.33 & & --- & 4.33 & & --- & 4.30\\
		Citeseer & 227,320 & 814,134 & & --- & 4.04 & & --- & 3.87 & & --- & 4.02\\
		MathSciNet & 332,689 & 820,644 & & --- & 4.32 & & --- & 4.37 & & --- & 4.27\\
		TwitterFollows & 404,719 & 713,319 & & --- & 3.91 & & --- & 3.94 & & --- & 3.86\\
		Flickr & 513,969 & 3,190,452 & & --- & 8.21 & & --- & 8.22 & & --- & 8.07\\
		Delicious & 536,108 & 1,365,961 & & --- & 5.48 & & --- & 5.39 & & --- & 5.30\\
		IMDB & 896,305 & 3,782,447 & & --- & 11.10 & & --- & 11.22 & & --- & 10.75\\
		YoutubeSnap & 1,134,890 & 2,987,624 & & --- & 9.00 & & --- & 8.89 & & --- & 8.98\\
		Lastfm & 1,191,805 & 4,519,330 & & --- & 11.79 & & --- & 11.79 & & --- & 12.12\\
		Pokec & 1,632,803 & 22,301,964 & & --- & 75.91 & & --- & 76.50 & & --- & 76.14\\
		Flixster & 2,523,386 & 7,918,801 & & --- & 19.46 & & --- & 19.67 & & --- & 19.56\\
		LiveJournal & 4,033,137 & 27,933,062 & & --- & 80.60 & & --- & 80.15 & & --- & 79.40\\			\bottomrule
	\end{tabular}
\end{table*}
\begin{table*}
	\centering
	\caption{Relative error for  estimated $\tilde{C}_{\rm I}(\calG)$,  $\tilde{D}(\calG)$,  $\tilde{P}(\calG)$, $\tilde{C}(\calG)$  for three internal distributions with input parameter $\epsilon=10^{-6}$.}\label{T2}
	\resizebox{460pt}{75pt}{
		\begin{tabular}{llcccclcccclcccc}
			\toprule
			\multirow{3}*{Network $\calG$} & & \multicolumn{14}{c}{Relative error of four estimated quantities for three internal opinion distributions ($\times 10^{-7}$)} \\
			\cmidrule{3-16}
			& & \multicolumn{4}{c}{Uniform distribution} & & \multicolumn{4}{c}{Exponential distribution} & &
			\multicolumn{4}{c}{Power-law distribution} \\
			\cmidrule{3-6}
			\cmidrule{8-11}
			\cmidrule{13-16}
			& & $\tilde{C}_{\rm I}(\calG)$ & $\tilde{D}(\calG)$ & $\tilde{P}(\calG)$ & $\tilde{I}_{\rm dc}(\calG)$ & & $\tilde{C}_{\rm I}(\calG)$ & $\tilde{D}(\calG)$ & $\tilde{P}(\calG)$ & $\tilde{I}_{\rm dc}(\calG)$ & & $\tilde{C}_{\rm I}(\calG)$ & $\tilde{D}(\calG)$ & $P(\calG)$ & $\tilde{I}_{\rm dc}(\calG)$ \\
			\midrule
			Erd\"{o}s992 & & 1.1002 & 0.0018 & 0.0018 & 0.1379 & & 6.5276 & 0.0077 & 0.0078 & 0.0903 & & 1.3976 & 0.0010 & 0.0009 & 0.0372\\
			Advogato & & 2.6516 & 0.0058 & 0.0077 & 0.0900 & & 2.0039 & 0.0138 & 0.0186 & 0.0827 & & 0.0975 & 0.0022 & 0.0025 & 0.0015\\
			PagesGovernment & & 0.6307 & 0.0214 & 0.0544 & 0.0432 & & 5.1476 & 0.0068 & 0.0165 & 0.0643 & & 0.0926 & 0.0020 & 0.0076 & 0.2168\\
			WikiElec & & 1.0029 & 0.0021 & 0.0030 & 0.0826 & & 4.9890 & 0.0002 & 0.0003 & 0.0582 & & 0.1049 & 0.0074 & 0.0220 & 0.0516\\
			HepPh & & 0.5886 & 0.0009 & 0.0015 & 0.2008 & & 0.6633 & 0.0039 & 0.0063 & 0.0605 & & 2.5963 & 0.0051 & 0.0098 & 0.0380\\
			Anybeat & & 0.1904 & 0.0152 & 0.0185 & 0.0366 & & 0.6107 & 0.0064 & 0.0077 & 0.1517 & & 0.0256 & 0.0044 & 0.0067 & 0.0532\\
			PagesCompany & & 0.6721 & 0.0090 & 0.0129 & 0.0038 & & 0.7049 & 0.0046 & 0.0069 & 0.0061 & & 0.8048 & 0.0099 & 0.0136 & 0.2874\\
			AstroPh & & 0.8058 & 0.0051 & 0.0109 & 0.1292 & & 2.9936 & 0.0020 & 0.0041 & 0.0873 & & 0.0143 & 0.0012 & 0.0120 & 0.0122\\
			CondMat & & 2.0614 & 0.0051 & 0.0088 & 0.1528 & & 2.1314 & 0.0027 & 0.0049 & 0.0670 & & 0.2099 & 0.0056 & 0.0155 & 0.0666\\
			Gplus & & 1.6386 & 0.0012 & 0.0014 & 0.2219 & & 0.4878 & 0.0001 & 0.0001 & 0.0198 & & 0.1471 & 0.0005 & 0.0006 & 0.0756\\
			GemsecRO & & 4.5966 & 0.0230 & 0.0358 & 0.3691 & & 2.1684 & 0.0069 & 0.0109 & 0.1042 & & 2.0774 & 0.0257 & 0.0507 & 0.4375\\
			GemsecHU & & 0.4021 & 0.0342 & 0.0826 & 0.2018 & & 4.6511 & 0.0153 & 0.0369 & 0.1320 & & 0.8842 & 0.0003 & 0.0004 & 0.0831\\
			PagesArtist & & 1.9994 & 0.0093 & 0.0228 & 0.4941 & & 1.3243 & 0.0077 & 0.0194 & 0.4454 & & 0.0419 & 0.0026 & 0.0125 & 0.0766\\
			Brightkite & & 1.2047 & 0.0022 & 0.0026 & 0.1436 & & 0.5782 & 0.0190 & 0.0221 & 0.1564 & & 0.5104 & 0.0421 & 0.0324 & 0.0027\\
			\bottomrule
		\end{tabular}
	}	
\end{table*}

\section{Experiments}

In this section, we assess the efficiency and accuracy of our approximation algorithm $\textsc{Approxim}$. To this end, we implement this algorithm on various real networks and  compare the running time and accuracy of $\textsc{ Approxim }$ with those corresponding to the exact algorithm, called \textsc{Exact}. For the  \textsc{Exact}, it computes relevant quantities by directly inverting matrix $\II+\LL$, performing product of related matrices, and then calculating corresponding  $\ell_2$ norms.

\textbf{Environment.} Our extensive  experiments  were  run on a Linux box  with 4-core 4.2GHz  Intel i7-7700K CPU and   32GB of main memory. Our code for both the approximation algorithm \textsc{Approxim} and  exact algorithm \textsc{Exact} was written in \textit{Julia v1.5.1}.  The  solver \textsc{Solve} we use is based on the technique in~\cite{KySa16},  the \textit{Julia language} implementation for which is open and accessible  on the website~\footnote{http://danspielman.github.io/Laplacians.jl/latest/}.

\textbf{Datasets.} All the real-world networks we consider are  publicly available in the   Koblenz Network Collection~\cite{Ku13}  and Network Repository~\cite{RyNe15}.   The first three columns of Table~\ref{T1} are related information  of networks, including  the network name, the number of nodes, and the number of edges. For a network with $n$ nodes and $m$ edges, we use $n'$ and $m'$ to denote, respectively, the number of nodes and edges in its largest connected component. All our  experiments were conducted on the largest connected components.  The smallest network consists of 4991 nodes, while the largest network  has more than one million. In Table~\ref{T1}, the networks listed in an increasing order of the number of nodes in their largest connected components.

\textbf{ Internal opinion distributions.} In our experiments, the internal opinions are generated according to three different distributions: uniform distribution, exponential distribution, and power-law distribution, with the latter two distributions generated according by the randht.py file in~\cite{ClShNe09}. For the uniform distribution, the opinion $s_i$ of node $i$ is distributed uniformly in the range of $[0,1]$.  For  the exponential distribution, we choose the   probability density $ \mathrm{e}^{ x_{\min }}e^{-x}$ to generate $n'$  positive numbers $x$ with minimum value $x_{\min}>0$. Then, dividing each $x$ by the maximum observed value, we  normalize  these $n'$ numbers to the range $[0,1]$ as the internal opinions of nodes. Note that  there is always a node with internal  opinion 1 due to the normalization operation. Similarly, for the power-law distribution, we use the  probability density $(\alpha-1)x_{\min}^{\alpha-1}x^{-\alpha}$ with $\alpha=2.5$ to generate $n'$ positive numbers, and normalize them to interval $[0,1]$ as the internal opinions.

\textbf{Efficiency.} Table~\ref{T1} reports the running time of the two algorithms \textsc{Approxim} and \textsc{Exact} on different networks we consider in the experiments. We note that  we cannot compute those related quantities using   the algorithm  \textsc{Exact} for the last 7 networks, due to the high memory and time cost, but we can run the algorithm   \textsc{Approxim}.  In all our experiments, we set the error parameter $\epsilon$ equal to $10^{-6}$. For each of the three internal opinion distributions in different networks, we record the running times of \textsc{Approxim} and \textsc{Exact}.  Table~\ref{T1} shows that for all considered networks, the running time for \textsc{Approxim} is less than that for \textsc{Exact}.  For moderately large networks with more than ten thousand nodes, \textsc{Exact} is much slower than \textsc{Approxim}. For example, for the GemsecRO network, the running time for \textsc{Approxim} is less than one hundredth of that for \textsc{Exact}. Finally, for large graphs particularly  those with over 150 thousand nodes, our approximation algorithm \textsc{Approxim} shows  a very obvious efficiency advantage, since \textsc{Exact} fails to run for these networks.


\textbf{Accuracy.} Except for the  high-efficiency, our approximation algorithm \textsc{Approxim} has also high-accuracy in practice. To demonstrate this, we  assess the accuracy of  algorithm \textsc{Approxim} in Table~\ref{T2}. For each of the three distributions of internal opinions, we compare the approximate results of \textsc{Approxim} with the exact result of \textsc{Exact} for all networks examined but the last seven ones in Table~\ref{T2}. For each quantity $\rho$ concerned, we use the mean relative error  $\sigma = | \rho - \tilde{\rho} | /\rho$ of   $\tilde{\rho}$ obtained by \textsc{Approxim} as an estimation of  $\rho$. In Table~\ref{T2}, we present  the mean relative errors of the four estimated quantities, internal conflict  $\tilde{C}_{\rm I}(\calG)$,  disagreement  $\tilde{D}(\calG)$,  polarization $\tilde{P}(\calG)$, and controversy  $\tilde{C}(\calG)$,   for different real networks with input error parameter $\epsilon=10^{-6}$. The results show that for all quantities concerned of all networks examined, the actual relative errors are ignorable, with all errors  less than $10^{-7}$,  much  smaller than the proved theoretical guarantee.  

\section{Conclusion}
Conflict, disagreement, controversy, and polarization in social networks have received considerable recent attention, and some indices have been proposed to quantify these concepts.  However, direct computation of these quantities is computationally challenging for large networks, because it involves inverting matrices. In this paper, we addressed the problem of efficiently computing these quantities in an undirected unweighted graph. To this end, we developed an approximation algorithm that is based on seminal work about Laplacian solver for solving linear system of equations, which can avoid the operation of matrix inversion. The algorithm has an almost linear computation complexity with respect to the number of edges in the graph, and simultaneously possesses a theoretical guarantee for accuracy. We performed extensive experiments on a diverse set of real network datasets to demonstrate that our presented approximation algorithm works both efficiently and effectively, especially for large-scale networks with millions of nodes.

\providecommand{\noopsort}[1]{}\providecommand{\singleletter}[1]{#1}%

\bibliographystyle{ACM-Reference-Format}


\end{document}